\documentclass[pra,twocolumn,floatfix,aps,showpacs,amsfonts,10pt]{revtex4}
\usepackage{graphicx}
\usepackage{array, color}
\usepackage{amsthm}
\usepackage{amsmath}
\usepackage{amssymb}
\newtheorem{theorem}{Theorem}

\newtheorem{definition}[theorem]{Definition}

\newtheorem{proposition}[theorem]{Proposition}

\begin{document}
\title{Epsilon-smooth measure of coherence}

\author{Zhengjun Xi}
\email{xizhengjun@snnu.edu.cn}
\affiliation{College of Computer Science, Shaanxi Normal University, Xi'an 710062,
China}
\author{Shanshan Yuwen}
\affiliation{College of Computer Science, Shaanxi Normal University, Xi'an 710062,
China}
\date{\today}

\begin{abstract}
In this paper, by minimizing the coherence quantifiers over all states in an $\epsilon$ ball around a given state, we define a generalized smooth quantifier, called the $\epsilon$-smooth measure of coherence. We use it to estimate the difference between the expected state and the actually prepared state and quantify quantum coherence contained in an actually prepared state, and it can been interpreted as the minimal coherence guaranteed to present in an $\epsilon$ ball around given quantum state.  We find that the $\epsilon$-smooth measure of any coherence monotone is still a coherence monotone, but it does not
satisfy monotonicity on average under incoherent operations. We show the $\epsilon$-smooth measure of coherence is continuous even if the original coherence quantifier is not. We also study the $\epsilon$-smooth measure of distance-based coherence quantifiers, and some interesting properties are given. Moreover, we discuss the dual form of the $\epsilon$-smooth measure of coherence by maximizing over all states in an $\epsilon$ ball around the given state and show that the dual $\epsilon$-smooth measure of coherence provides an upper bound of one-shot coherence distillation.
\end{abstract}
%\eid{identifier}
%\pacs{}
\maketitle
\section{Introduction}
Quantum resource theory is an extraordinary framework to build a rigorous quantification of distinct features of
quantum information theory~\cite{BrabdaoPRL15,CoeckeIC16,RegulaJPA18,Chitambararxiv18}.
The coherent superposition is an essential ingredient for entanglement and for several quantum features,
and it does not present in classical physics~\cite{VedralPRL97,PlenioQIC07,HorodeckiRMP09}. The resource theory of quantum coherence is a vivid
research topic~\cite{Aberg06,BaumgratzPRL14}, and related theories and various applications in this direction
have also been investigated~\cite{AdessoJPA16,StreltsovRMP17,HuPR18}.

In entanglement theory, the $\epsilon$ measure of entanglement
depends on a precision parameter $\epsilon$, and it quantifies the entanglement
contained in a state which is only partially known~\cite{MoraNJP08}, and it and its variant have been applied and further developed
~\cite{FernandoCMP10,MoraPRA10,BuscemiPRL11,NestPRL13,DattaIEEEIT09a,DattaIEEEIT09b,KoingIEEEIT09,DattaIEEEIT11,BuscemiIEEEIT13,WildeIEEEIT17}. We know that the coherence embodies the essence of entanglement in the multipartite system, that is to say, we can also associate to every coherence measure a family of measures which depend on parameter $\epsilon$. The $\epsilon$ version of every coherence quantifier may be actually considered as a smoothed version of the coherence quantifiers. Thus, we can consider $\epsilon$ smooth for every coherence measure by either minimizing or maximizing the measure over all states in an $\epsilon$ ball around a given state.
There are two major motivations for studying the $\epsilon$-smooth measures of coherence.

First, we need to prepare the high-quality coherence states in order to achieve useful quantum information processing tasks, for example, quantum algorithms, quantum metrology, and so on.
In fact, we know that any preparation apparatus has realistically only a certain degree of precision and reliability, that is to say,
there is a certain distance between the expected state and the actually prepared state. Then, we can estimate their difference and quantify the coherence contained in actually prepared state. The $\epsilon$-smooth measure of coherence can provide a bound of the coherence prepared in a state with some threshold values $\epsilon$.
If one takes the minimum over all states in an $\epsilon$ ball, then the $\epsilon$-smooth measure of coherence aims to characterize the minimum guaranteed amount of coherence,
given the promises that the state which has actually been prepared is within a distance $\epsilon$ from the expected state.

Another reason for studying the $\epsilon$ measure of coherence is that not long after Baumgratz, Cramer, and Plenio proposed the resource theory of coherence~\cite{BaumgratzPRL14}, Winter and Yang established an operational theory of coherence by introducing the coherence distillation and the coherence cost (coherence dilution)~\cite{WinterPRL16}.
In general, the distillable coherence cannot be larger than the coherence cost,
since the resources are finite and the number of
independent and identically distributed prepared states is necessarily limited.
More importantly, it is very hard to perform coherent state manipulations over a very large numbers of systems.
Therefore, it becomes crucial to be able to characterize
how well one can distill maximally coherent states from prepared states, or how well one can dilute a unit resource state to the target state.
The study of such nonasymptotic and especially one-shot scenarios has garnered great interest in quantum coherence, in which one has access only to a single copy of a quantum system and
allows for a finite accuracy, reflecting the realistic restrictions on state transformations. Zhao $et$ $al$. established the one-shot theory of coherence dilution~\cite{ZhaoPRL18}, and Regula $et$ $al$. characterized the distillation of quantum coherence in the one-shot setting~\cite{RegulaPRL18}. Recently, Zhao $et$ $al$. gave a full review for the problem of one-shot coherence
distillation under four classes of incoherent operations~\cite{Winter18}.
During the same period, several authors also discuss the one-shot scenarios of assisted distillable coherence~\cite{HsiehJPA18,RegulaPRA18}.
We know that the smooth minimum and maximum relative entropies of coherence
play a significant role in the one-shot coherence
distillation and cost~\cite{ZhaoPRL18,RegulaPRL18,Winter18,HsiehJPA18,RegulaPRA18},
In particular, Bu $et$ $al$. investigated the minimum and maximum relative entropies
and gave an operational interpretation of the maximum relative entropy of coherence~\cite{BuPRL17}.

The paper is organized as follows. We give the definition of the $\epsilon$-smooth measure of coherence by minimizing the coherence measures over all states in an $\epsilon$ ball around a given state in Sec.~\ref{sec:definition},
and in Sec.~\ref{sec:prop} we discuss the properties of the $\epsilon$-smooth measure of coherence. In Sec.~\ref{sec:distance}, we provide the lower and upper bounds which establish a relation with the distance-based coherence quantifiers. We discuss the dual form of the $\epsilon$-smooth measure of coherence in Sec.~\ref{sec:MAX}. Sec.~\ref{sec:conclusion} is devoted to our conclusions.

\section{The Definitions of $\epsilon$-smooth measures of coherence}\label{sec:definition}
We briefly give an account of the concepts that are required to derive our main results.
We are concerned with the resource theory of coherence by described in Refs.~\cite{BaumgratzPRL14,StreltsovRMP17}.
Consider a finite $d$-dimensional
Hilbert space $\mathcal{H}$ with a fixed reference basis $\{|i\rangle\}_{i=1}^d$, in which the set of incoherent
states $\mathcal{I}$ is defined as the set of all the states of the form,
\begin{equation}
\delta=\sum_i\delta_i|i\rangle\langle i|,
\end{equation}
where $\delta_i$ are probabilities, and $\sum_i\delta_i=1$. Let $\mathcal{D}$ be the convex set of density operators acting on Hilbert space $\mathcal{H}$; we then have $\mathcal{I}\subset\mathcal{D}$.
Any state which cannot be written as above is defined as a coherent state,
which means the coherence is basis dependent.
Baumgratz \emph{et al.} proposed that any proper measure of the coherence $\mathcal{C}$ must satisfy the following conditions~\cite{BaumgratzPRL14}:
\begin{itemize}
\item [{(C1)}]  $Non-negativity:$ $\mathcal{C}(\rho)\geqslant 0$ for all quantum states $\rho$, and $\mathcal{C}(\rho)=0$ if and only if $\rho\in \mathcal{I}$.
\item [{(C2)}]  $Monotonicity:$ $\mathcal{C}(\rho)$ is nonincreasing under incoherent operation $\Lambda$, i.e.,
$\mathcal{C}(\rho) \geqslant \mathcal{C}(\Lambda(\rho))$.
\item [{(C3)}]  $Strong$ $monotonicity:$ $\mathcal{C}(\rho)$ is nonincreasing on average under selective incoherent operations, i.e., $\mathcal{C}(\rho)\geqslant \sum_k p_k\mathcal{C}(\rho_k) $, where $\rho_k=K_k\rho K_k^\dag/p_k$ and $p_k=\mathrm{Tr}(K_k \rho K_k^\dag)$ for all $\{K_k\}$ with $\sum_k K_k^{\dagger}K_k= I$ and $K_k \mathcal{I} K_k^\dagger\subseteq\mathcal{I}$.
\item [{(C4)}]  $Convexity:$ $\mathcal{C}(\rho)$ is a convex function of quantum states, i.e.,
%\begin{equation}
$\sum_ip_i \mathcal{C}(\rho_i)\geqslant \mathcal{C}(\sum_ip_i\rho_i)$
%\end{equation}
 for any ensemble $\{p_i,\rho_i\}$.
\end{itemize}

Note that conditions (C3) and (C4) automatically
imply condition (C2).
%The condition (C3) is important because it allows for sub-selection
%based on measurement outcomes, a process available in well controlled by quantum experiments.
It has been shown that the relative entropy and $l_1$-norm satisfy all conditions~\cite{BaumgratzPRL14}.
%However, the measures of coherence induced by the squared Hilbert-Schmidt norm satisfies conditions (C1), (C2), (C4), but not (C3).
Recently, we have found that the measure of coherence induced by the fidelity
does not satisfy condition (C3), and an explicit example is presented~\cite{XiPRA15}. In general, if a quantity $\mathcal{C}$ fulfills condition
(C1) and either condition (C2) or  (C3) (or both), we say this is a coherence monotone, and if a quantity $\mathcal{C}$
fulfills conditions (C1)-(C4), we say this is a coherence measure. For more discussion we refer to two reviews~\cite{HuPR18,StreltsovRMP17}.

We know that the distance $D$ is a very important tool in metric theory, and it is applied in many respects in information theory~\cite{Nielsenbook,Wildebook}. The distance needs to satisfy non-negative, symmetric, and triangle inequality. We know that some functions do not satisfy the fundamental conditions of the distance, but they are still widely used in information theory, e.g. the relative entropy. In this paper, we require that the distance $D$ is also convex or jointly convex, that is,
\begin{equation}
D\left(\sum_i\lambda_i\rho_i, \sum_i\lambda_i\tau_i\right)\leqslant \sum_i\lambda_iD(\rho_i,\tau_i),
\end{equation}
where $\lambda_i\geq 0$ and $\sum_i\lambda_i=1$.
 Furthermore, we also require that the distance $D$ is contractive under completely positive and trace-preserving operations~$\Phi$, that is,
 \begin{equation}
 D(\Phi(\rho),\Phi(\tau))\leqslant D(\rho,\tau).
 \end{equation}
 Clearly, the trace distance meets our all requirements~\cite{Nielsenbook,Wildebook}. It is defined by
 \begin{equation}
 D_{tr}(\rho,\tau)=\frac{1}{2}\mathrm{Tr}|\rho-\tau|,
 \end{equation}
 where $|X|=\sqrt{X^\dagger X}$.
With the definition of coherence measure at hand, we present the definition of the $\epsilon$ smooth to all coherence measures by minimizing over all states in an $\epsilon$ ball around $\rho$ as follows.

\begin{definition}\label{epsilon-m-coh}
For every coherence measure $\mathcal{C}$, and any $\epsilon\geqslant 0$, and fixed distance $D$, the $\epsilon$-smooth measures of coherence are defined as
\begin{equation}
\mathcal{C}_{\epsilon,\min}^{(D)}(\rho)=\min_{\tau\in B_{\epsilon}^{(D)}(\rho)}\mathcal{C}(\tau),
\end{equation}
where $B_{\epsilon}(\rho)=\{\tau|D(\rho,\tau)\leqslant\epsilon\}$.
\end{definition}

Note that the distance and the measure function of coherence may be the same, or they can also be different.
Henceforth, we will omit the superscript $D$ to avoid unnecessary misunderstands.
Intuitively, the $\epsilon$-smooth measure of coherence $\mathcal{C}_{\epsilon,\min}$ describes the minimal amount of coherence guaranteed to present in an $\epsilon$ ball around the given quantum state $\rho$. This construction gives a technique to obtain a smooth
function even when the original measure $\mathcal{C}$ is not. This is consistent with the entanglement case~\cite{MoraNJP08} and it is similar to the smooth version of coherence in the one-shot coherence scenarios~\cite{BuPRL17,ZhaoPRL18,RegulaPRL18,Winter18,HsiehJPA18,RegulaPRA18}.

Obviously, for $\epsilon=0$, they reduce to the nonsmooth cases, we have
$\mathcal{C}_{0,\min}(\rho)=\mathcal{C}(\rho)$.
 If the parameter $\epsilon$ is so large that there is a incoherent state in $B_{\epsilon}(\rho)$, we find that $\mathcal{C}_{\epsilon,\min}(\rho)$ vanishes on all the set of states $\mathcal{D}$. In general, for a fixed coherent state $\rho$, we are typically interested in values of $\epsilon$ that are
smaller than the distance of the state $\rho$ from the set of coherent states, that is, $\epsilon\leqslant D(\rho,\delta)$, where $\delta\in \mathcal{I}$.
Thus, we can give an order relation of the $\epsilon$-smooth measure of coherence via the different parameters.
\begin{proposition}\label{pror:order_p}
For any $\epsilon>\epsilon^\prime\geqslant0$,
we have
\begin{equation}
0\leqslant\mathcal{C}_{\epsilon^\prime,\min}(\rho)-\mathcal{C}_{\epsilon,\min}(\rho)\leqslant \frac{\epsilon-\epsilon^\prime}{\epsilon}\mathcal{C}(\rho).
\end{equation}
\end{proposition}
\begin{proof}
Since $\epsilon>\epsilon^\prime\geqslant0$, then we have $B_{\epsilon^\prime}(\rho)\subseteq B_{\epsilon}(\rho)$, this implies that
\begin{equation}\label{eq:bound_ccc}
\mathcal{C}_{\epsilon,\min}(\rho)\leqslant\mathcal{C}_{\epsilon^\prime,\min}(\rho).
\end{equation}
Then we will prove the rightmost hand of inequality.
Suppose that the state $\tau\in B_{\epsilon}(\rho)$ such that $\mathcal{C}_{\epsilon,\min}(\rho)=\mathcal{C}(\tau)$, then we denote $\lambda=\epsilon^\prime/\epsilon$ and take the state $\tau_{\lambda}=\lambda\tau+(1-\lambda)\rho$. Since the distance $D$ is jointly convex, then we obtain
\begin{equation}
D(\rho,\tau_\lambda)\leqslant \lambda D(\rho,\tau)\leqslant \epsilon^\prime.
\end{equation}
This implies that the state $\tau_\lambda\in B_{\epsilon^\prime}(\rho)$. Further, we have
\begin{align}
\mathcal{C}_{\epsilon^\prime,\min}(\rho)&\leqslant \lambda\mathcal{C}(\tau)+(1-\lambda)\mathcal{C}(\rho)\nonumber\\
&\leqslant \mathcal{C}_{\epsilon,\min}(\rho)+(1-\lambda)\mathcal{C}(\rho).
\end{align}
Thus we have
\begin{equation}\label{eq:bound_cccc}
\mathcal{C}_{\epsilon^\prime,\min}(\rho)-\mathcal{C}_{\epsilon,\min}(\rho)\leqslant \frac{\epsilon-\epsilon^\prime}{\epsilon}\mathcal{C}(\rho).
\end{equation}
Combining Eq.~(\ref{eq:bound_ccc}) with Eq.~(\ref{eq:bound_cccc}), we obtain the desired result.
\end{proof}
\section{Basic Properties}\label{sec:prop}
In this section, we will list some basic properties of the $\epsilon$-smooth measure of coherence.
\begin{proposition}\label{pror:order_p2}
For any quantum state $\rho$, $\mathcal{C}_{\epsilon,\min}(\rho)$ is non-increasing under incoherent operation $\Lambda$, i.e.,
\begin{equation}
\mathcal{C}_{\epsilon,\min}(\rho) \geqslant \mathcal{C}_{\epsilon,\min}(\Lambda(\rho)).
\end{equation}
\end{proposition}
\begin{proof}
Suppose that the state $\tau^*\in B_{\epsilon}(\rho)$ realizes the minimum for the associated $\epsilon$ measure $\mathcal{C}_{\epsilon,\min}(\rho)$, i.e.,
\begin{equation}\label{sec:P_eq1}
\mathcal{C}_{\epsilon,\min}(\rho)=\mathcal{C}(\tau^*).
\end{equation}
Since the distance $D$ is contractive under incoherent operation $\Lambda$, it follows that
\begin{equation}
D(\Lambda(\rho),\Lambda(\tau^*))\leqslant D(\rho,\tau^*)\leqslant \epsilon.
\end{equation}
Thus we have
\begin{align}
\mathcal{C}(\tau^*)&\geqslant\mathcal{C}(\Lambda(\tau^*))\nonumber\\
&\geqslant \min_{\hat{\tau}\in B_{\epsilon}(\Lambda(\rho))}\mathcal{C}(\hat{\tau})\nonumber\\
&=\mathcal{C}_{\epsilon,\min}(\Lambda(\rho)).
\end{align}
From Eq.~(\ref{sec:P_eq1}), we obtain the desired result.
\end{proof}
From Propositions~\ref{pror:order_p} and~\ref{pror:order_p2}, we know that the $\epsilon$-smooth minimum measure of coherence $\mathcal{C}_{\epsilon,\min}$ satisfies conditions (C1) and (C2), which implies that the $\epsilon$-smooth measure of coherence $\mathcal{C}_{\epsilon,\min}$ is a coherence monotone.
But the following result shows that this measure does not satisfy condition (C3).
\begin{proposition}
The $\epsilon$-smooth measure of coherence $\mathcal{C}_{\epsilon,\min}$ does not satisfy strong monotonicity .
\end{proposition}
\begin{proof}
Let us suppose a state $\rho$, and
\begin{equation}
\rho=\eta|0\rangle\langle 0|\otimes \rho_{c}+(1-\eta)|1\rangle\langle 1|\otimes \delta,
\end{equation}
where $\eta\in(0,1]$, $|0\rangle$ and $|1\rangle$ are orthogonal states of a local qubit system. Here, we require that $\mathcal{C}_{\epsilon,\min}(\rho_{c})>0$. Note that the state $\delta$ is an incoherent state, where we have $\mathcal{C}_{\epsilon,\min}(\delta)=0$.
Without loss of generality, if one considers the trace distance, we have $D_{tr}(\rho,|1\rangle\langle 1|\otimes \delta)=\eta$.
Thus we say that the parameter $\eta$ is small enough so that the incoherent state $|1\rangle\langle 1|\otimes \delta\in B_{\epsilon}(\rho)$, which implies that the following equality holds, i.e.,
 \begin{equation}\label{eq:strong_m1}
\mathcal{C}_{\epsilon,\min}(\rho)=0.
\end{equation}
Then one can perform an incoherent operation with the Kraus elements $|0\rangle\langle 0|\otimes I$ and $|1\rangle\langle 1|\otimes I$ on the state $\rho$, the two outputs of such measurement are $|0\rangle\langle 0|\otimes \rho_{c}$ and $|1\rangle\langle 1|\otimes \delta $ with probabilities $\eta$ and $1-\eta$, respectively. Then, we have
\begin{align}\label{eq:strong_m2}
\eta\mathcal{C}&_{\epsilon,\min}(|0\rangle\langle 0|\otimes\rho_{c})+(1-\eta)\mathcal{C}_{\epsilon,\min}(|1\rangle\langle 1|\otimes\delta)\nonumber\\
&=\eta\mathcal{C}_{\epsilon,\min}(\rho_{c})+(1-\eta)\mathcal{C}_{\epsilon,\min}(\delta)\nonumber\\
&=\eta\mathcal{C}_{\epsilon,\min}(\rho_{c}).
\end{align}
The above equality together with the fact the $\mathcal{C}_{\epsilon,\min}(\rho_{c})>0$ and Eq.~(\ref{eq:strong_m1}) implies that we obtain our desired result.
\end{proof}

Note that the first equality in Eq.~(\ref{eq:strong_m2}) arises from the following relation, i.e.,
\begin{equation}
\mathcal{C}_{\epsilon,\min}(|0\rangle\langle 0|\otimes\rho)=\mathcal{C}_{\epsilon,\min}(\rho).
\end{equation}
In fact, for any two states $|0\rangle\langle 0|\otimes\tau$ and $\tau^{\prime}\otimes\tau$ in $B_{\epsilon}(|0\rangle\langle 0|\otimes\rho))$, we have
\begin{equation}
\mathcal{C}(\tau)=\mathcal{C}(|0\rangle\langle 0|\otimes\tau)\leqslant\mathcal{C}(\tau^{\prime}\otimes\tau).
\end{equation}
Thus we have
\begin{align}
\mathcal{C}_{\epsilon,\min}(|0\rangle\langle 0|\otimes\rho)=&\min_{\widetilde{\tau}\in B_{\epsilon}(|0\rangle\langle 0|\otimes\rho)}C(\widetilde{\tau})\nonumber\\
=&\min_{\tau\in B_{\epsilon}(\rho)}C(|0\rangle\langle 0|\otimes\tau)\nonumber\\
=&\min_{\tau\in B_{\epsilon}(\rho)}C(\tau)\nonumber\\
=&\ \mathcal{C}_{\epsilon,\min}(\rho).
\end{align}
In general, we can not obtain that the $\epsilon$-smooth measure of coherence $\mathcal{C}_{\epsilon}$ satisfies the additive, i.e.,
\begin{equation}
\mathcal{C}_{\epsilon,\min}(\rho\otimes\sigma)=\mathcal{C}_{\epsilon,\min}(\rho)+\mathcal{C}_{\epsilon,\min}(\sigma).
\end{equation}

From the above result, according to the resource theory of coherence, although we have known that the $\epsilon$-smooth measure of coherence $\mathcal{C}_{\epsilon,\min}$ is not a proper coherence measure, the following result shows that it satisfies the convexity.
\begin{proposition}
The $\epsilon$ measure of coherence $\mathcal{C}_{\epsilon,\min}$ is convex, i.e.,
\begin{equation}
\mathcal{C}_{\epsilon,\min}\left(\sum_i\lambda_i\rho_i\right)\leqslant \sum_i\lambda_i\mathcal{C}_{\epsilon,\min}(\rho_i),
\end{equation}
where $\lambda_i\geqslant 0$ and $\sum_i\lambda_i=1$.
\end{proposition}
\begin{proof}
Suppose that the states $\tau^*_i\in B_{\epsilon}(\rho_i)$ realize the minimum in $\mathcal{C}_{\epsilon,\min}(\rho_i)$, $i=1,\cdots, n$. For every $i$, then we have
\begin{equation}
\mathcal{C}_{\epsilon}(\rho_i)=\mathcal{C}(\tau^*_i).
\end{equation}
Since the distance $D$ is joint convexity, we have
\begin{equation}
D\left(\sum_i\lambda_i\rho_i,\sum_i\lambda_i\tau^*_i\right)\leqslant \sum\lambda_iD(\rho_i,\tau^*_i)\leqslant \epsilon.
\end{equation}
This implies that the state $\sum_i\lambda_i\tau^*_i\in B_{\epsilon}(\sum_i\lambda_i\rho_i)$.
Thus we have
\begin{align}
\mathcal{C}_{\epsilon,\min}\left(\sum_i\lambda_i\rho_i\right)\leqslant & \ \mathcal{C}\left(\sum_i\lambda_i\tau^*_i\right)\nonumber\\
\leqslant& \ \sum_i\lambda_i\mathcal{C}(\tau^*_i)\nonumber\\
=& \ \sum_i\lambda_i\mathcal{C}_{\epsilon,\min}(\rho_i).
\end{align}
\end{proof}

The $\epsilon$-smooth measure of coherence $\mathcal{C}_{\epsilon,\min}$ is not a proper coherence measure form the axiomatic requirement in~\cite{BaumgratzPRL14},
but the convex property holds such that it is still a good coherence monotone. We also know strong monotonic measures that are convex will satisfy the condition (C2) and do not need to considered separately. In the following we will prove that other possibly relevant
properties of the original quantity $\mathcal{C}$ hold for $\mathcal{C}_{\epsilon,\min}$ too.
\begin{proposition}
The $\epsilon$-smooth measure of coherence $\mathcal{C}_{\epsilon,\min}$ is continuous in $\epsilon$ and $\rho$, for all $
\epsilon$ and for all $\rho$.
\end{proposition}
\begin{proof}
From Proposition~\ref{pror:order_p}, for fixed $\rho$, we can directly prove that,
\begin{equation}
|\mathcal{C}_{\epsilon,\min}(\rho)-\mathcal{C}_{\epsilon^\prime,\min}(\rho)|\rightarrow 0,
\end{equation}
as $\epsilon^\prime\rightarrow \epsilon$.

Then, we want to prove that for any $\epsilon>0$, and for any state $\rho$,
\begin{equation}
|\mathcal{C}_{\epsilon,\min}(\rho^\prime)-\mathcal{C}_{\epsilon,\min}(\rho)|\rightarrow 0,
\end{equation}
as $D(\rho,\rho^\prime)\rightarrow \epsilon$, for fixed $\epsilon>0$, and $\tau\in B_{\epsilon}(\rho)$ such that $\mathcal{C}_{\epsilon,\min}(\rho)=\mathcal{C}(\tau)$. We denote $\eta=D(\rho^\prime,\rho)$, $\lambda=\epsilon/(\epsilon+\eta)$, and take the state $\tau_\lambda=\lambda\tau+(1-\lambda)\rho^\prime$, since the distance $D$ is convex, and then we have
\begin{align}
D(\rho^\prime,\tau_\lambda)&\leqslant\lambda D(\rho^\prime,\tau)\nonumber\\
&\leqslant \lambda(D(\rho^\prime,\rho)+D(\rho,\rho^\prime))\nonumber\\
&\leqslant\lambda(\epsilon+\eta)=\epsilon.
\end{align}
This implies that the state $\tau_\lambda\in B_\epsilon(\rho^\prime)$, and we have
\begin{align}
\mathcal{C}_{\epsilon,\min}(\rho^\prime)&\leqslant\lambda\mathcal{C}(\tau)+(1-\lambda)\mathcal{C}(\rho^\prime)\nonumber\\
&=\lambda\mathcal{C}_{\epsilon}(\rho)+(1-\lambda)\mathcal{C}_{\epsilon}(\rho^\prime).
\end{align}
Thus we have
\begin{equation}
\mathcal{C}_{\epsilon,\min}(\rho^\prime)-\mathcal{C}_{\epsilon,\min}(\rho)\leqslant \frac{\eta}{\epsilon+\eta}(\mathcal{C}(\rho^\prime,\min)-\mathcal{C}_{\epsilon,\min}(\rho)).
\end{equation}
By exchanging the role of $\rho$ and $\rho^\prime$, we have
\begin{equation}
-\frac{\eta}{\epsilon+\eta}(\mathcal{C}(\rho)-\mathcal{C}_{\epsilon,\min}(\rho^\prime))
\leqslant\mathcal{C}_{\epsilon,\min}(\rho^\prime)-\mathcal{C}_{\epsilon,\min}(\rho)
\end{equation}
Taking $M=\max\{\mathcal{C}(\rho)-\mathcal{C}_{\epsilon,\min}(\rho^\prime),
\mathcal{C}(\rho^\prime)-\mathcal{C}_{\epsilon,\min}(\rho)\}$, we have
\begin{equation}
|\mathcal{C}_{\epsilon,\min}(\rho^\prime)-\mathcal{C}_{\epsilon,\min}(\rho)|\leqslant\frac{\eta}{\epsilon+\eta}M.
\end{equation}
Thus we obtain our desired result.
\end{proof}
This shows that the $\epsilon$-smooth measure of coherence $\mathcal{C}_{\epsilon,\min}$ is always a continuous function of the parameter $\epsilon$ and the state $\rho$, even though the original quantity $\mathcal{C}$ is noncontinuous or no one can presently prove it is continuous, (e.g., Robustness of coherence~\cite{NapoliPRL16} and the coherence number~\cite{XiarXiv18,ChinPRA17},  to our knowledge, no one can prove that robustness of coherence is continuous.). This result holds for any choice of distance $D$ and for all
bounded and convex coherence monotone. This contrasts with the case of entanglement~\cite{MoraNJP08}. We say that a significant advantage of the $\epsilon$ generalization of quantum resource monotones is to allow one to transform noncontinuous quantity into continuous ones.

\section{Distance-based coherent measures}\label{sec:distance}
Let us consider the family of distance-based coherence quantifier, as introduced in~\cite{BaumgratzPRL14,StreltsovRMP17,HuPR18}, defined by
\begin{equation}\label{def:distance coherence}
\mathcal{C}_D(\rho)=\inf_{\delta\in\mathcal{I}}D(\rho,\delta),
\end{equation}
where the infimum is taken over the set of incoherent states $\mathcal{I}$. Clearly, any quantity defined in~(\ref{def:distance coherence}) fulfills condition (C1) for any distance $D$ which is non-negative and zero and only if $\rho=\delta$. The condition (C2) is also satisfied if the distance $D$ is contractive. Any distance-based coherence quantifier fulfills condition (C4) whenever the corresponding distance $D$ is jointly convex. There are already three important distance-based coherence quantifiers: the $l_1$ norm  of coherence~\cite{BaumgratzPRL14}, geometric coherence via the fidelity~\cite{BaumgratzPRL14,XiPRA15,StreltsovPRL15}, AND THE trace distance of coherence~\cite{BaumgratzPRL14,RanaPRA16}.

We know that the definitions of incoherent operations or free operations are not unique and there are different choices~\cite{StreltsovRMP17}. Here we can define a particular incoherent operation $\Lambda^\delta_p$ with any incoherent state $\delta$ and the parameter $p\in[0,1]$, that is,
\begin{equation}\label{def:incoherent operation11}
\Lambda^\delta_p(\rho)=(1-p)\rho+p\delta.
\end{equation}
With this particular incoherent operation at hand, we present an explicit evolutional relation of distance-based coherence quantifier.
\begin{proposition}\label{prop:coh_D_incho_coh_D}
Suppose the distance $D$ is convex and contractive. Given any state $\rho$ and probability $p$, if $\delta^*$ realizes the optimal in~(\ref{def:distance coherence}), then \begin{equation}
\mathcal{C}_D(\Lambda^{\delta^*}_p(\rho))=(1-p)\mathcal{C}_D(\rho).
\end{equation}
\end{proposition}
\begin{proof}
By the convexity of the distance $D$, we have
\begin{align}\label{eq:E_R_coh11}
\mathcal{C}_D(\Lambda^{\delta^*}_p(\rho))&=\inf_{\delta\in\mathcal{I}}D(\Lambda^{\delta^*}_p(\rho),\delta)\nonumber\\
&\leqslant \inf_{\delta\in\mathcal{I}}\left[(1-p)D(\rho,\delta)+p D(\delta^*,\delta)\right]\nonumber\\
&\leqslant(1-p)\mathcal{C}_D(\rho).
\end{align}
On the other hand, by the triangle inequality and the convexity of the distance $D$, we have
\begin{align}\label{eq:E_R_coh}
\mathcal{C}_D(\Lambda^{\delta^*}_p(\rho))&=\inf_{\delta\in\mathcal{I}}D(\Lambda^{\delta^*}_p(\rho),\delta)\nonumber\\
&\geqslant\inf_{\delta\in\mathcal{I}}\left[D(\rho,\delta)-D(\rho,\Lambda^{\delta^*}_p(\rho))\right]\nonumber\\
&\geqslant \inf_{\delta\in\mathcal{I}}D(\rho,\delta)-p D(\rho,\delta^*)\nonumber\\
&=(1-p)\mathcal{C}_D(\rho).
\end{align}
Thus we obtain our desired result.
\end{proof}
We are now in a position to present lower and upper bounds for the $\epsilon$-smooth measure of coherence $\mathcal{C}_{\epsilon}$
via the
original measure $\mathcal{C}$ and the distance-based measure $\mathcal{C}_D$.
\begin{proposition}\label{Prop:lower upper bound}
Let $\mathcal{C}$ be a coherence measure and the distance $D$ be a convex and contractive, then the $\epsilon$ measure of coherence $\mathcal{C}_{\epsilon}$ satisfies
\begin{equation}
\min_{\substack{\tau\  \mathrm{s. t.} \\ \mathcal{C}_D(\tau)=\mathcal{C}_D(\rho)-\epsilon}}\mathcal{C}(\tau)\!\leqslant\!\mathcal{C}_{\epsilon,\min}(\rho)
\!\leqslant\!\left(1\!-\!\frac{\epsilon}{\mathcal{C}_D(\rho)}\right)\mathcal{C}(\rho),
\end{equation}
where $\mathcal{C}_D(\rho)$ is the coherence quantifier associated to the distance $D$.
\end{proposition}
\begin{proof}
By definition, for any state $\tau\in B_{\epsilon}(\rho)$,
we have
\begin{equation}
\mathcal{C}_{\epsilon,\min}(\rho)\leqslant \mathcal{C}(\tau).
\end{equation}Since the coherence measure $\mathcal{C}$ and the distance $D$ are convex, with the incoherent operation~(\ref{def:incoherent operation11}), we have
\begin{equation}
\mathcal{C}(\Lambda^{\delta}_p(\rho))\leqslant (1-p)\mathcal{C}(\rho)
\end{equation}
 and
\begin{equation}
D(\Lambda^{\delta}_p(\rho),\rho)\leqslant pD(\rho,\delta).
\end{equation}
Since we require $D(\rho,\delta)\geqslant\epsilon$, then one can choice $p$ such that $p\leqslant \epsilon/D(\rho,\delta)$ and we have
$D(\Lambda^{\delta}_p(\rho),\rho)\leqslant\epsilon$. Therefore we obtain
\begin{align}
\mathcal{C}_{\epsilon,\min}(\rho)&\leqslant \min_{\substack{\Lambda^{\delta}_p s. t. \\ D(\Lambda^{\delta}_p(\rho),\rho)\leq \epsilon}}\mathcal{C}(\Lambda^{\delta}_p(\rho))\nonumber\\
&\leqslant\min_{\substack{\Lambda^{\delta}_p s. t. \\ D(\Lambda^{\delta}_p(\rho),\rho)\leq \epsilon}}(1-p)\mathcal{C}(\rho)\nonumber\\
&\leqslant\min_{\delta}\left(1-\frac{\epsilon}{D(\rho,\delta)}\right)\mathcal{C}(\rho)\nonumber\\
&\leqslant \left(1-\frac{\epsilon}{\mathcal{C}_D(\rho)}\right)\mathcal{C}(\rho),
\end{align}
where the minimum in the first and second inequalities are taken over all the incoherent states $\delta$ and the parameter $p$ via the condition $D(\Lambda^{\delta}_p(\rho),\rho)\leqslant \epsilon$, and the third inequality comes from a restriction of the minimum to the case where we fix the parameter $p=\epsilon/D(\rho,\delta)$.

Next we prove the lower bound. Suppose that the state $\rho^\prime\in B_{\epsilon}(\rho)$ realizes the minimum for the quantifier $\mathcal{C}_{\epsilon}$. We then have
\begin{equation}
\mathcal{C}_{\epsilon,\min}(\rho)=\mathcal{C}(\rho^\prime)\geqslant \mathcal{C}(\Lambda^{\delta}_p(\rho^\prime)).
\end{equation}
Without loss of generality, we assume that the incoherent state $\delta^*$ is optimal for $\mathcal{C}_{D}(\rho^\prime)$, and then we obtain
\begin{align}
\mathcal{C}_{D}(\rho^\prime)&=D(\rho^\prime,\delta^*)\nonumber\\
&\geqslant D(\rho,\delta^*)-D(\rho,\rho^\prime)\nonumber\\
&\geqslant \mathcal{C}_{D}(\rho)-\epsilon.
\end{align}
We can take the parameter $s=1-(\mathcal{C}_{D}(\rho)-\epsilon)/\mathcal{C}_{D}(\rho^\prime)$.
Obviously, we have $0\leqslant s\leqslant 1$. From Proposition~\ref{prop:coh_D_incho_coh_D}, we obtain
\begin{equation}
\mathcal{C}_{D}(\Lambda_s^{\delta^*}(\rho^\prime))
=(1-s)\mathcal{C}_{D}(\rho^\prime)=\mathcal{C}_{D}(\rho)-\epsilon.
\end{equation}
Therefore we have
\begin{align}
\mathcal{C}_{\epsilon,\min}(\rho)&=\mathcal{C}(\rho^\prime)\nonumber\\
&\geqslant\mathcal{C}(\Lambda_s^{\delta^*}(\rho^\prime))\nonumber\\
&\geqslant \min_{\substack{\tau\  \mathrm{s. t.} \\ \mathcal{C}_D(\tau)=\mathcal{C}_D(\rho)-\epsilon}}\mathcal{C}(\tau).
\end{align}
Thus we obtain our desired result.
\end{proof}
Note that the lower and upper bound depend on the fundamental properties of the distance, e.g., the symmetric and the triangle inequality. In particular, if one considers the distance-based coherence quantifier $\mathcal{C}_{D}$, from Proposition~\ref{Prop:lower upper bound}, we give the relation between the $\epsilon$-smooth measure of coherence $\mathcal{C}_{\epsilon,\min}$ and the coherence quantifier $\mathcal{C}_{D}$, namely,
\begin{equation}
\mathcal{C}_{D}(\rho)=\mathcal{C}_{\epsilon,\min}(\rho)+\epsilon.
\end{equation}

Note that the relative entropy is not a proper distance because it does not satisfy symmetric and triangle inequality, but the relative entropy of coherence~\cite{BaumgratzPRL14} is also viewed as a distance-based coherence quantifiers, and there are many interesting properties~\cite{BaumgratzPRL14,StreltsovRMP17, WinterPRL16,BuPRL17,XiSR15}. Next we will discuss separately the $\epsilon$ generalization of
relative entropy of coherence. We first introduce the definition of the relative entropy of coherence~\cite{BaumgratzPRL14},
which is defined by
\begin{equation}
\mathcal{C}_{r}(\rho)=\min_{\delta\in\mathcal{I}}S(\rho||\delta),
\end{equation}
where $S(\tau||\delta)=\mathrm{Tr}\rho(\log_2 \rho-\log_2\delta)$ is the relative entropy, and $supp(\rho)\subseteq supp(\sigma)$.
Then we can rewrite definition~\ref{epsilon-m-coh} via the relative entropy as follows.
\begin{definition}\label{epsilon-m-coh-re}
For the relative entropy of coherence $\mathcal{C}_r$, and any $\epsilon\geqslant 0$,
we define
\begin{equation}\label{def:epsilon_REC}
\mathcal{C}_{r,\epsilon}(\rho)=\min_{\tau\in B_{\epsilon}(\rho)}\mathcal{C}_r(\tau),
\end{equation}
where $B_{\epsilon}(\rho)=\{\tau|S(\rho||\tau)\leqslant \epsilon\}$, and we call it the $\epsilon$-smooth relative entropy of coherence.
\end{definition}
We hope that there does not exist the incoherent state in the $\epsilon$ ball around a given coherent state $\rho$; thus, for any state $\tau\in B_{\epsilon}(\rho)$,
we require that $S(\rho||\tau)\leqslant\epsilon\leqslant\mathcal{C}_r(\tau)$.
Clearly, we know that the $\epsilon$-smooth relative entropy of coherence is non-negative and convex,
and it is also nonincreasing under incoherent operation $\Lambda$.
Since the relative entropy of coherence is continuous, it is easy to verify that the $\epsilon$-smooth relative entropy of coherence is also continuous.
We do not prove strong monotonicity, but we can give a weak strong monotonicity.

Suppose that the state $\tau\in B_{\epsilon}(\rho)$ realizes the minimum for the quantifier $\mathcal{C}_{r,\epsilon}(\rho)$. We then have
\begin{equation}\label{eq:SM_1111}
\mathcal{C}_{r,\epsilon}(\rho)=\mathcal{C}_{r}(\tau).
\end{equation}
Let $\rho_k=K_k\rho K_k^\dag/p_k$ with $p_k=\mathrm{Tr}(K_k \rho K_k^\dag)$, and $\tau_k=K_k\tau K_k^\dag/q_k$ with $q_k=\mathrm{Tr}(K_k \tau K_k^\dag)$ for all Kraus operators $\{K_k\}$ with $\sum_k K_k^{\dagger}K_k= I$ and $K_k \mathcal{I} K_k^\dagger\subseteq\mathcal{I}$. From the properties of relative entropy~\cite{VedralPRA98}, we have
\begin{equation}
\sum_kp_kS(\rho_k||\tau_k)\leqslant S(\rho||\tau).
\end{equation}
Since the state $\tau\in B_{\epsilon}(\rho)$, that is, $S(\rho||\tau)\leqslant\epsilon$, then we could require $S(\rho_k||\tau_k)\leqslant\epsilon$ for all $k$,
which implies that the states $\tau_k\in B_{\epsilon}(\rho_k)$. From the definition of the $\epsilon$-smooth relative entropy of coherence, we have
\begin{equation}
\mathcal{C}_{r,\epsilon}(\rho_k)\leqslant\mathcal{C}_{r}(\tau_k).
\end{equation}
Since the relative entropy of coherence is nonincreasing on average under selective incoherent operations, we then have
\begin{equation}\label{eq:SM_2222}
\sum_kq_k\mathcal{C}_{r,\epsilon}(\rho_k)\leqslant\sum_kq_k\mathcal{C}_{r}(\tau_k)\leqslant \mathcal{C}_{r}(\tau).
\end{equation}
Combining Eq.~(\ref{eq:SM_1111}) with Eq.~(\ref{eq:SM_2222}), we obtain
\begin{equation}
\sum_kq_k\mathcal{C}_{r,\epsilon}(\rho_k)\leqslant\mathcal{C}_{r,\epsilon}(\rho).
\end{equation}

Note that if the probabilities $q_k$ are replaced with the probabilities $p_k$, this is a real strong monotonicity. This is the reason we call it weak strong monotonicity.
In particular, if one considers the incoherent operation~(\ref{def:incoherent operation11}), for the $\epsilon$-measure relative entropy of coherence, we have
 \begin{equation}
\mathcal{C}_{r, \epsilon}(\Lambda^{\delta}_p(\rho))\leqslant(1-p)\mathcal{C}_{r,\epsilon}(\rho).
\end{equation}
Further, using the jointly convexity of relative entropy, we have
 \begin{equation}
S(\rho||\Lambda^{\delta}_p(\rho))\leqslant pS(\rho||\delta).
\end{equation}
If we choose the parameter $p=\epsilon/D(\rho,\delta)$, from Proposition~\ref{Prop:lower upper bound} we obtain
 \begin{equation}
\mathcal{C}_{r, \epsilon}(\rho)\leqslant\mathcal{C}_{r}(\rho)-\epsilon.
\end{equation}

\section{The Dual form of The $\epsilon$-smooth measure of coherence}\label{sec:MAX}

We have discussed the smoothed coherence quantifiers with the minimal value in an $\epsilon$ ball, in this section we will consider the maximal value case.
This form is dual to the $\epsilon$-smooth measure of coherence in Eq.~(\ref{epsilon-max-coh}) as the minimum over states in $\epsilon$ ball. This motivation comes from the smoothed minimum and maximum relative entropies of coherence (or entanglement), where there are very significant applications in the one-shot coherence (or entanglement) distillation and dilution
~\cite{BuPRL17,ZhaoPRL18,RegulaPRL18,HsiehJPA18,RegulaPRA18,Winter18,DattaIEEEIT09a,DattaIEEEIT09b,KoingIEEEIT09,DattaIEEEIT11,BuscemiIEEEIT13,WildeIEEEIT17}.
%and we refer to Ref.~\cite{} for a comprehensive review of recent developments in the one-shot coherence distillation problem.
\begin{definition}\label{epsilon-max-coh}
For every coherence measure $\mathcal{C}$, and any $\epsilon\geqslant 0$ and fixed distance $D$, the dual $\epsilon$-smooth measure of coherence is defined as
\begin{equation}
\mathcal{C}_{\epsilon,\max}(\rho)=\max_{\tau\in B_{\epsilon}(\rho)}\mathcal{C}(\tau).
\end{equation}
\end{definition}

Obviously, for $\epsilon=0$, we have $
\mathcal{C}_{0,\max}(\rho)=\mathcal{C}_{0,\min}(\rho)=\mathcal{C}(\rho)$.
 If the parameter $\epsilon$ is so large that there is a maximally coherent state in $B_{\epsilon}(\rho)$, we are likely to get
\begin{equation}
\mathcal{C}_{\epsilon,\max}(\rho)=\log_2d.
\end{equation}
This means that the dual $\epsilon$-smooth measure of coherence is not a coherence monotone, as it indeed violated condition (C1). For a given coherent state $\rho$, we have $\mathcal{C}_{0,\max}(\rho)> 0$. But the converse direction may not be true. Suppose that $|0\rangle$ and $|1\rangle$ are a fixed orthonormal basis in a qubit system, and let $|a\rangle=|0\rangle$ and $|b\rangle=\cos\theta|0\rangle+\sin\theta|1\rangle$, where $\theta\in(0,\frac{\pi}{2})$. Clearly, we have $\mathcal{C}(|a\rangle)=0$ and $\mathcal{C}(|b\rangle)>0$.
Since we have $D_{tr}(|a\rangle,|b\rangle)=|\sin\theta|$, for a given $\epsilon$, by some choices of $\theta$, we have $|\sin\theta|\leqslant\epsilon$. This implies that $|b\rangle\in B_{\epsilon}(|a\rangle)$. Therefore we have
\begin{equation}
\mathcal{C}_{\epsilon,\max}(|a\rangle)\geqslant\mathcal{C}(|b\rangle)>0.
\end{equation}

In addition, from the definition of the dual $\epsilon$-smooth measure of coherence, we cannot prove that this quantifier is monotone under the incoherent operations.
For every incoherent operation $\Lambda$, if one takes
\begin{equation}
B_{\epsilon}(\Lambda(\rho))=\{\Lambda(\tau)|\tau\in B_{\epsilon}(\rho)\},
\end{equation}
within this limitation, we obtain the following result.
\begin{proposition}\label{Hyp porp}
For any incoherent operation $\Lambda$, we have
\begin{equation}\label{Hyp cond}
\mathcal{C}_{\epsilon,\max}(\rho)\geqslant \mathcal{C}_{\epsilon,\max}(\Lambda(\rho)).
\end{equation}
\end{proposition}
\begin{proof}
Setting $\Lambda(\tau^*)$ as the optimal state in the definition of $\mathcal{C}_{\epsilon,\max}(\Lambda(\rho))$, we have
\begin{align}
\mathcal{C}_{\epsilon,\max}(\Lambda(\rho))&=\mathcal{C}(\Lambda(\tau^*))\nonumber\\
&\leqslant \mathcal{C}(\tau^*)\nonumber\\
&\leqslant \max_{\tau\in B_{\epsilon}(\rho)}\mathcal{C}(\tau)\nonumber\\
&=\mathcal{C}_{\epsilon,\max}(\rho).
\end{align}
In the first inequality we used that the coherence measure $\mathcal{C}$ is monotone under the incoherent operation $\Lambda$, and the second inequality comes from our hypothesis~(\ref{Hyp cond}), that is, $\tau^*\in B_{\epsilon}(\rho)$.
\end{proof}

%Thus, we are not expecting to do that restrictions such that the dual $\epsilon$-smooth measure of coherence is a coherence measure.
The question we were interested in is whether its various forms can play a significant role in the one-shot coherence distillation. Intuitively, the dual $\epsilon$-smooth measure of coherence $\mathcal{C}_{\epsilon,\max}$ is generalization of the $\epsilon$-smooth minimum relative entropy of coherence~\cite{BuPRL17,ZhaoPRL18,RegulaPRL18,HsiehJPA18,RegulaPRA18,Winter18}. If one takes the coherence quantifier $\mathcal{C}$ by the minimum relative entropy and considers the purified distance, then the dual $\epsilon$-smooth measure of coherence reduces to the smoothed minimum relative entropy of coherence~\cite{Winter18}. This quantity has some interesting properties in the manipulation of coherence. It provided an upper bound of one-shot coherence distillation, and in the asymptotic limit the smoothed minimum relative entropy of coherence is equivalent to the relative entropy of coherence~\cite{BuPRL17}.
Then the next definition formalizes the notion of the one-shot coherence distillation under the fixed distance $D$ and any coherence measure $\mathcal{C}$.
\begin{definition}\label{def_one-shot CD}
The one-shot coherence distillation of $\rho$ under the fixed distance $D$, any coherence measure $\mathcal{C}$, and a class of quantum incoherent operations $\Lambda$ is defined as
\begin{equation}
\mathcal{C}_{d,\Lambda}^{(1),\epsilon}(\rho)=\max_{\Lambda}\{c_M:D(\Lambda(\rho),\Psi_M)\leqslant\epsilon\},
\end{equation}
where $\Psi_M$ is a maximally state rank $M$ on Hilbert space, and $|\Psi_M\rangle=\frac{1}{\sqrt{M}}\sum_i|i\rangle$, and $c_M=\mathcal{C}(\Psi_M)$.
\end{definition}

Remarkably, our Definition~\ref{def_one-shot CD} is a generalization of the one-shot coherence distillation in~\cite{BuPRL17,RegulaPRL18,Winter18}. The below result relates the dual $\epsilon$-smooth measure of coherence to the one-shot coherence distillation under quantum incoherent operations. Here, we require quantum incoherent operations satisfy the condition in Proposition~\ref{Hyp porp}.
\begin{proposition}
For any state $\rho$ and any $\epsilon\geqslant 0$, we have
\begin{equation}
\mathcal{C}_{d,\Lambda}^{(1),\epsilon}(\rho)\leqslant \mathcal{C}_{\epsilon,\max}(\rho).
\end{equation}
\end{proposition}
\begin{proof}
Suppose $\Lambda^*$ is the optimal incoherent operation such that $D(\Lambda^*(\rho),\Psi_M)\leqslant\epsilon$ with $c_M=\mathcal{C}(\Psi_M)=\mathcal{C}_{d,\Lambda}^{(1),\epsilon}(\rho)$, then we have $\Psi_M\in B_{\epsilon}(\Lambda^*(\rho))$. From Proposition~\ref{Hyp porp} we have
\begin{align}
\mathcal{C}_{\epsilon,\max}(\rho)&\geqslant \mathcal{C}_{\epsilon,\max}(\Lambda^*(\rho))\nonumber\\
&=\max_{\Lambda^*(\tau)\in B_{\epsilon}(\Lambda^*(\rho))}\mathcal{C}(\Lambda^*(\tau))\nonumber\\
&\geqslant\mathcal{C}(\Psi_M)\nonumber\\
&=\mathcal{C}_{d,\Lambda}^{(1),\epsilon}(\rho).
\end{align}
Thus we obtain our desired result.
\end{proof}
This shows that the dual $\epsilon$-smooth measure of coherence $\mathcal{C}_{\epsilon,\max}$ can estimate the maximum amount of maximally coherent states created for a reliable coherence distillable protocol.

Similarly, we find that the $\epsilon$-smooth measure of coherence gives a lower bound of the one-shot coherence cost under quantum incoherence operations. We define the one-shot coherence cost of a quantum state $\rho$ as follows.
\begin{definition}\label{def_one-shot CC}
The one-shot coherence cost of $\rho$ under the fixed distance $D$, any coherence measure $\mathcal{C}$, and quantum incoherence operation $\Lambda$ is defined as
\begin{equation}
\mathcal{C}_{c,\Lambda}^{(1),\epsilon}(\rho)=\min_{\Lambda}\{c_M:D(\rho,\Lambda(\Psi_M))\leqslant\epsilon\}.
\end{equation}
\end{definition}
Remarkably, our Definition~\ref{def_one-shot CC} is a generalization of the one-shot coherence cost in~\cite{BuPRL17,ZhaoPRL18}, where they consider the fidelity as a distance. Then the following result relates the $\epsilon$-smooth measure of coherence to the one-shot coherence coherence under quantum incoherent operations.
\begin{proposition}
For any state $\rho$ and any $\epsilon\geqslant 0$, we have
\begin{equation}
\mathcal{C}_{\epsilon,\min}(\rho)\leqslant\mathcal{C}_{c,\Lambda}^{(1),\epsilon}(\rho).
\end{equation}
\end{proposition}
\begin{proof}
Suppose $\Lambda^*$ is the optimal incoherent operation such that $D(\rho,\Lambda^*(\Psi_M))\leqslant\epsilon$ with $c_M=\mathcal{C}(\Psi_M)=\mathcal{C}_{c,\Lambda}^{(1),\epsilon}(\rho)$. This means that $\Lambda^*(\Psi_M)\in B_{\epsilon}(\rho)$, and then we have
\begin{align}
\mathcal{C}_{\epsilon,\min}(\rho)& \leqslant \mathcal{C}(\Lambda^*(\Psi_M))\nonumber\\
& \leqslant \mathcal{C}(\Psi_M)\nonumber\\
&=\mathcal{C}_{c,\Lambda}^{(1),\epsilon}(\rho).
\end{align}
Thus we obtain our desired result.
\end{proof}
This shows that the $\epsilon$-smooth measure of coherence $\mathcal{C}_{\epsilon,\min}$ provides a lower bound of the minimum amount of maximally coherent states needs for faithful coherence cost (or dilution).
\section{Conclusions}\label{sec:conclusion}
We have introduced the $\epsilon$-smooth measure of coherence via the original coherence quantifier.
The $\epsilon$-smooth measure of coherence of a given quantum state could have been interpreted as the
minimum guaranteed coherence contained in the actually prepared state when we
only know that it is $\epsilon$ close to the ideal target state. We have shown that the $\epsilon$-smooth measure of any coherence monotone is
still a coherence monotone. We have proved that the $\epsilon$-smooth measure of coherence quantifier does not
satisfy monotonicity on average under incoherent operations. In particular, we found that the $\epsilon$ measure of relative entropy of coherence
satisfies a weak monotonicity on average under incoherent operations.
We also proved that the $\epsilon$ measure of a convex coherence monotone is continuous, and thus it can also be seen as a smooth version of the original non-continuous quantity. We have given a lower and upper bound of the $\epsilon$-smooth measure of distance-based coherence quantifier. We have discussed the dual form of the $\epsilon$-smooth measure of coherence and have shown that it provides an upper bound of the one-shot coherence distillation. Similarly, we also got that the $\epsilon$-smooth measure of coherence gives an lower bound of the one-shot coherence cost.

We believe that the results obtained will help to understand quantum resource theory. As potential next steps, our results could be extended to an infinite-dimensional system or continuous setting. We can also discuss the regularized version of the $\epsilon$-smooth measure of coherence and its applications in quantum information processing.
\section{Acknowledgments}
We acknowledge X.-M. Lu for the insightful discussions.
Z. Xi is supported by the National Natural Science Foundation of China (Grants No. 61671280, No. 11531009, and No, 11771009), the Natural
Science Basic Research Plan in Shaanxi Province of China
(No. 2017KJXX-92),
 the Fundamental Research Funds for the
Central Universities, and by the Funded Projects for the Academic Leaders and Academic Backbones, Shaanxi Normal University (16QNGG013).

\end{document}